\documentclass[submission,copyright,creativecommons]{eptcs}
\providecommand{\event}{ICE 2015}

\usepackage{amsthm}
\usepackage{amsmath}
\usepackage{amssymb}
\usepackage{times}
\usepackage{txfonts}
\usepackage{stmaryrd}
\usepackage{url}
\usepackage{amsfonts}
\usepackage{code}
\usepackage{mathrsfs}
\DeclareMathAlphabet{\mathpzc}{OT1}{pzc}{m}{it}
\let\mathpzc\mathscr
\let\mathpzc\mathcal
\def\BNF{\ \  | \ \  }
\def\bondi{{\bf bondi}}
\def\ltrans{[\![}
\def\rtrans{]\!]}

\usepackage{prooftree}
\usepackage{macros}

\newtheorem{theorem}{Theorem}[section]
\newtheorem{lemma}[theorem]{Lemma}
\newtheorem{proposition}[theorem]{Proposition}
\newtheorem{corollary}[theorem]{Corollary}
\newtheorem{remark}[theorem]{Remark}

\newenvironment{definition}[1][Definition]{\begin{trivlist}
\item[\hskip \labelsep {\bfseries #1}]}{\end{trivlist}}

\renewcommand{\beq}{\simeq}

\makeatletter
\def \rightarrowfill{\m@th\mathord{\smash-}\mkern-6mu%
  \cleaders\hbox{$\mkern-2mu\mathord{\smash-}\mkern-2mu$}\hfill
  \mkern-6mu\mathord\to}
\makeatother
\makeatletter
\def \Rightarrowfill{\m@th\mathord{\smash=}\mkern-6mu%
  \cleaders\hbox{$\mkern-2mu\mathord{\smash=}\mkern-2mu$}\hfill
  \mkern-6mu\mathord\Rightarrow}
\makeatother
\def \overstackrel#1#2{\mathrel{\mathop{#1}\limits^{#2}}}

\newcommand{\ifte}[4]{{\bf if}\ #1=#2\ {\bf then}\ #3\ {\bf else}\ #4}
\newcommand{\ift}[3]{{\bf if}\ #1=#2\ {\bf then}\ #3}
\newcommand{\join}[1]{(#1)\rhd }

\newcommand{\cood}[1]{\,\mbox{\sc Cd}(#1)}

\title{On the Expressiveness of Joining}

\author{%
Thomas Given-Wilson
\institute{Inria, France}
\and
Axel Legay
\institute{Inria, France}
}

\def\titlerunning{On the Expressiveness of Joining}
\def\authorrunning{T. Given-Wilson \& A. Legay}

\begin{document}
\makeatactive

\maketitle  

\begin{abstract}
The expressiveness of communication primitives has been explored in a common framework
based on the $\pi$-calculus by considering
four features:
{\em synchronism} (asynchronous vs synchronous),
{\em arity} (monadic vs polyadic data),
{\em communication medium} (shared dataspaces vs channel-based),
and {\em pattern-matching} (binding to a name vs testing name equality  vs intensionality).
Here another dimension {\em coordination} is considered that accounts for the number of
processes required for an interaction to occur.
Coordination generalises binary languages such as $\pi$-calculus to {\em joining} languages
that combine inputs such as the Join Calculus and general rendezvous calculus.
By means of possibility/impossibility of encodings, this paper shows 
coordination is unrelated to the other features.
That is, joining languages are more expressive than binary languages, and no combination
of the other features can encode a joining language into a binary language.
Further, joining is not able to encode any of the other features unless they could be encoded
otherwise.
\end{abstract}

\section{Introduction}

The expressiveness of process calculi based upon their choice of communication primitives
has been explored before
\cite{Palamidessi:2003:CEP:966707.966709,journals/iandc/BusiGZ00,DeNicola:2006:EPK:1148743.1148750,G:IC08,GivenWilsonPHD,givenwilson:hal-01026301}.
In \cite{G:IC08} and \cite{givenwilson:hal-01026301} this is detailed by examining combinations
of four features, namely:
{\em synchronism}, asynchronous versus synchronous;
{\em arity}, monadic versus polyadic;
{\em communication medium}, shared dataspaces versus channels;
and {\em pattern-matching}, purely binding names versus name equality versus intensionality.
These features are able to represent many popular calculi \cite{G:IC08,givenwilson:hal-01026301} such as:
asynchronous or synchronous,  monadic or polyadic $\pi$-calculus
  \cite{Milner:1992:CMP:162037.162038,Milner:1992:CMP:162037.162039,milner:polyadic-tutorial};
\Linda~\cite{Gel85};
Mobile Ambients  \cite{DBLP:conf/fossacs/CardelliG98};
$\mu${\sc Klaim} \cite{10.1109/32.685256};
semantic-$\pi$ \cite{Castagna:2008:SSP:1367144.1367262};
and asymmetric concurrent pattern calculus \cite{DBLP:journals/corr/Given-Wilson14}.
Also the intensional features capture significant aspects of
Concurrent Pattern Calculus (CPC) \cite{GivenWilsonGorlaJay10,givenwilson:hal-00987578}
and variations \cite{GivenWilsonPHD,DBLP:journals/corr/Given-Wilson14};
and Psi calculi \cite{BJPV11} and sorted Psi calculi \cite{DBLP:conf/tgc/BorgstromGPVP13}.

Typically interaction in process calculi is a binary relation, where two processes interact
and reduce to a third process. For example in $\pi$-calculus the interaction rule is
\begin{eqnarray*}
\oap m a .P \bnf \iap m x .Q & \redar & P \bnf \{a/x\} Q\; .
\end{eqnarray*}
Here the processes $\oap m a .P$ and $\iap m x .Q$ interact and reduce to a new process $P \bnf \{a/x\}Q$.
However, there are process calculi that are not binary with their interactions.
For example, Concurrent Constraint Programming (CCP) has no direct interaction primitives, instead interactions are between a single process and the constraint environment \cite{Saraswat:1991:SFC:99583.99627}.
In the other direction Join Calculus \cite{Fournet_thereflexive}, general rendezvous calculus \cite{Bocchi2004119}, and m-calculus \cite{DBLP:conf/popl/SchmittS03}
allow any number of processes to join in a single interaction. 

This paper abstracts away from specific calculi in the style of \cite{G:IC08,givenwilson:hal-01026301}
to provide a general account of the expressiveness of the {\em coordination} of
communication primitives.
Here coordination can be either {\em binary} between an explicit input and output (as above),
or {\em joining} where the input may interact with unbounded outputs (but at least one).
For example, consider the reduction
\begin{eqnarray*}
\oap m a . P_1 \bnf \oap n b . P_2 \bnf \join{\iap m x \bnf \iap n y}Q
&\redar& P_1\bnf P_2\bnf \{a/x,b/y\}Q
\end{eqnarray*}
where the join $\rhd$ interacts when the two outputs $\oap m a$ and $\oap n b$ can
match the two parts of the input $\iap m x$ and $\iap n y$, respectively.

By adding the dimension of coordination, the original 24 calculi of
\cite{G:IC08,givenwilson:hal-01026301}
are here expanded to 48.
This paper details the relations between these calculi, with the following key results.

Joining cannot be encoded into a binary language. This is formalised via the {\em coordination degree}
of a language that is the least upper bound on the number of processes required to yield a reduction.
In general a language with a greater coordination degree cannot be encoded into a language with a
lesser coordination degree. That is, the joining languages with $\infty$ coordination degree cannot be
encoded into the binary languages with coordination degree 2.

Joining synchronous languages can be encoded into joining asynchronous languages when their binary
counterparts allow an encoding from a synchronous language into an asynchronous one.
In the other direction synchronous languages cannot be encoded into
asynchronous languages that differ only by the addition of joining over binary communication.

Polyadic languages that cannot be encoded into monadic languages in the binary setting cannot be
encoded into monadic languages simply with the addition of joining. Indeed, coordination is
unrelated to arity despite being similar in having a base case (monadic/binary) and an unbounded
case (polyadic/joining).

Channel-based languages cannot be encoded into dataspace-based languages by the addition of joining
unless they could be encoded already. In the other direction, the addition of channels does not allow
a joining language to be encoded into a binary language.

Intensionality cannot be encoded by joining regardless of other features, this result mirrors
the general result that intensionality cannot be represented by any combination of the first four
features \cite{givenwilson:hal-01026301}.
Name-matching cannot be encoded by joining into a
language without any name-matching, despite the possibility of matching unbounded numbers of
names via joining on an unbounded number of channels.

Overall, the results of this paper prove that joining is orthogonal to all the other features,
and that joining languages are strictly more expressive than binary languages.

The structure of the paper is as follows.
Section~\ref{sec:calculi} introduces the 48 calculi considered here.
Section~\ref{sec:encoding} revises the criteria used for encoding and comparing calculi.
Section~\ref{sec:join_only} defines the coordination degree of a language and formalises the
relation between binary and joining languages.
Section~\ref{sec:join_synch} considers the relation between synchronism and coordination.
Section~\ref{sec:join_arity} relates arity and coordination.
Section~\ref{sec:join_comm} presents results contrasting communication medium with coordination.
Section~\ref{sec:join_pattern} formalises the relation between pattern-matching and coordination.
Section~\ref{sec:conclude} concludes, discusses future and related work, and provides some motivations
for intensional calculi.

\section{Calculi}
\label{sec:calculi}

This section defines the syntax, operational, and behavioural semantics of the calculi
considered here. This relies heavily on the well-known notions developed for the
$\pi$-calculus (the reference framework) and adapts them when necessary to cope with
different features. With the exception of the joining constructs this is a repetition
of prior definitions from \cite{givenwilson:hal-01026301}.

Assume a countable set of names ${\mathcal N}$ ranged over by $a,b,c,\ldots$. Traditionally in 
$\pi$-calculus-style calculi names are used for channels, input bindings, and output data. However, here these are generalised to account for structure. Then, define the {\em terms} (denoted with $s,t,\ldots$) to be
$s,t ::= a\BNF s\bullet t$.
Terms consist of names such as $a$, or of {\em compounds} $s\bullet t$ that combines two terms into one.
The choice of the $\bullet$ as compound operator is similar to Concurrent Pattern Calculus, and also to be clearly distinct from the traditional comma-separated tuples of polyadic calculi.

The input primitives of different languages will exploit different kinds of {\em input patterns}.
The non-pattern-matching languages will simply use binding names (denoted $x,y,z,\ldots$).
The {\em name-matching} patterns, denoted $m,n,o,\ldots$ and defined by
$m,n ::= x\BNF \pro a$
consist of either a {\em binding name} $x$, or a {\em name-match} $\pro a$.
Lastly the {\em intensional patterns} (denoted $p,q,\ldots$) will also consider structure and are defined by
$p,q ::= m\BNF p\bullet q$.
The binding names $x$ and name-match $\pro a$ are contained in $m$ from the name-matching calculi,
the {\em compound pattern} $p\bullet q$ combines $p$ and $q$ into a single pattern, and is left associative.
The free names and binding names of name-matching and intensional patterns are as expected, taking
the union of sub-patterns for compound patterns. Note that an intensional pattern is well-formed
if and only if all binding names within the pattern are pairwise distinct.
The rest of this paper will only consider well-formed intensional patterns.

The (parametric) syntax for the languages is:
\begin{eqnarray*}
P,Q,R &::=& \zero\BNF OutProc \BNF InProc \BNF \res a P \BNF P|Q\BNF \ifte s t P Q \BNF *P\BNF\ok\; .
\end{eqnarray*}
The different languages are obtained by replacing the output $OutProc$ and input $InProc$ with the various definitions.
The rest of the process forms are as usual:
$\zero$ denotes the null process;
restriction $\res a P$ restricts the visibility of $a$ to $P$;
and parallel composition $P|Q$ allows independent evolution of $P$ and $Q$.
The $\ifte s t P Q$ represents conditional equivalence with $\ift s t P$ used when $Q$ is $\zero$.
The~$*P$ represents replication of the process $P$.
Finally, the $\ok$ is used to represent a success process or state, exploited for reasoning about
encodings as in \cite{G:CONCUR08,GivenWilsonPHD}.

This paper considers the possible combinations of five features for communication:
{\em synchronism} (asynchronous vs synchronous),
{\em arity} (monadic vs polyadic data),
{\em communication medium} (dataspace-based vs channel-based),
{\em pattern-matching} (simple binding vs name equality vs intensionality),
and {\em coordination} (binary vs joining).
As a result there exist 48 languages denoted as $\Lambda_{s,a,m,p,b}$ whose generic element is denoted as $\Lang_{\alpha,\beta,\gamma,\delta,\epsilon}$ where:
\begin{itemize}
\item $\alpha = A$ for asynchronous communication, and $\alpha = S$ for synchronous communication.
\item $\beta = M$ for monadic data, and $\beta = P$ for polyadic data.
\item $\gamma = D$ for dataspace-based communication, and $\gamma = C$ for channel-based communications.
\item $\delta = \mathit{NO}$ for no matching capability, $\delta = \mathit{\mathit{NM}}$ for name-matching, and $\delta = I$ for intensionality.
\item $\epsilon = B$ for binary communication, and $\epsilon = J$ for joining communication.
\end{itemize}
For simplicity a dash $-$ will be used when the instantiation of that feature is unimportant.

\begin{figure}[t]
\begin{equation*}
\begin{array}{rcll}
\Lang_{A,-,-,-,-}:&&OutProc ::= OUT\\
\Lang_{S,-,-,-,-}:&&OutProc ::= OUT.P\quad\ \ \\
\Lang_{-,-,-,-,B}:&&InProc ::= IN.P\\
\Lang_{-,-,-,-,J}:&&InProc ::= \join{\mathcal{I}}P&\mathcal{I}::=IN \ \ \bnf\ \ \mathcal{I}\bnf\mathcal{I}\\
\Lang_{-,M,D,\mathit{NO},-}:&&IN ::= \iap {} x& OUT::= \oap {} a\\
\Lang_{-,M,D,\mathit{NM},-}:&&IN ::= \iap {} m& OUT::= \oap {} a\\
\Lang_{-,M,D,I,-}:&&IN ::= \iap {} p& OUT::= \oap {} t\\
\Lang_{-,M,C,\mathit{NO},-}:&&IN ::= \iap a x& OUT::= \oap a b\\
\Lang_{-,M,C,\mathit{NM},-}:&&IN ::= \iap a m& OUT::= \oap a b\\
\Lang_{-,M,C,I,-}:&&IN ::= \iap s p& OUT::= \oap s t\\
\Lang_{-,P,D,\mathit{NO},-}:&&IN ::= \iap {} {\wt x}& OUT::= \oap {} {\wt a}\\
\Lang_{-,P,D,\mathit{NM},-}:&&IN ::= \iap {} {\wt m}& OUT::= \oap {} {\wt a}\\
\Lang_{-,P,D,I,-}:&&IN ::= \iap {} {\wt p}& OUT::= \oap {} {\wt t}\\
\Lang_{-,P,C,\mathit{NO},-}:&&IN ::= \iap a {\wt x}& OUT::= \oap a {\wt b}\\
\Lang_{-,P,C,\mathit{NM},-}:&&IN ::= \iap a {\wt m}& OUT::= \oap a {\wt b}\\
\Lang_{-,P,C,I,-}:&\qquad&IN ::= \iap s {\wt p}\qquad& OUT::= \oap s {\wt t}\vspace{-0.4cm}
\end{array}
\end{equation*}
\caption{Syntax of Languages.}
\label{fig:syntax}
\vspace{-0.3cm}
\end{figure}

Thus the syntax of every language is obtained from the productions in Figure~\ref{fig:syntax}.
The denotation $\wt \cdot$ represents a sequence of the form $\cdot _1,\cdot _2,\ldots,\cdot _n$ and can be used for names, terms, and input patterns.

As usual $\iap a {\ldots,x,\ldots}.P$ and $\res x P$ and $\iap {} {x\bullet\ldots} . P$ and $\join {\ldots\bnf\iap a x \bnf \ldots} P$ bind $x$ in $P$.
Observe that in $\iap a {\ldots,\pro b,\ldots} .P$ and $\iap {} {\ldots\bullet\pro b} .P$
neither $a$ nor $b$ bind in $P$, both are free.
The corresponding notions of free and bound names of a process, denoted ${\sf fn}(P)$ and ${\sf bn}(P)$,
are as usual.
Also note that $\alpha$-equivalence, denoted $=_\alpha$ is assumed in the usual manner.
Lastly, an input is well-formed if all binding names in that input occur exactly once. This paper shall only consider well-formed inputs.
Finally, the structural equivalence relation $\equiv$ is defined by: 
\begin{equation*}
\begin{array}{c}
P\bnf \zero \equiv P
\qquad\qquad
P\bnf Q \equiv Q \bnf P
\qquad\qquad
P\bnf (Q\bnf R) \equiv (P\bnf Q)\bnf R
\vspace{0.2cm} \\
\ifte s t P Q \equiv P\quad s = t
\qquad\qquad
\ifte s t P Q \equiv Q\quad s\neq t
\vspace{0.2cm} \\
P\equiv P'\quad\mbox{if}\ P=_\alpha P'
\qquad\qquad
\res a \zero \equiv \zero
\qquad \res a \res b P\equiv \res b \res a P
\vspace{0.2cm} \\
P\bnf \res a Q\equiv \res a (P\bnf Q)\quad \mbox{if}\ a\notin{\sf fn}(P)
\qquad\qquad
*P\equiv P\bnf *P \; .\vspace{-0.4cm}
\end{array}
\end{equation*}

Observe that $\Lang_{A,M,C,\mathit{NO},B}$, $\Lang_{A,P,C,\mathit{NO},B}$, $\Lang_{S,M,C,\mathit{NO},B}$, and $\Lang_{S,P,C,\mathit{NO},B}$
align with the communication
primitives of the asynchronous/synchronous monadic/polyadic $\pi$-calculus
\cite{Milner:1992:CMP:162037.162038,Milner:1992:CMP:162037.162039,milner:polyadic-tutorial}.
The language $\Lang_{A,P,D,\mathit{NM},B}$ aligns with \Linda \cite{Gel85};
the languages $\Lang_{A,M,D,\mathit{NO},B}$ and $\Lang_{A,P,D,\mathit{NO},B}$ with the monadic/polyadic Mobile Ambients \cite{DBLP:conf/fossacs/CardelliG98};
and $\Lang_{A,P,C,\mathit{NM},B}$ with that of $\mu${\sc Klaim} \cite{10.1109/32.685256} or semantic-$\pi$ \cite{Castagna:2008:SSP:1367144.1367262}.
The intensional languages do not exactly match any well-known calculi.
However,
the language $\Lang_{S,M,D,I,B}$ has been mentioned in \cite{GivenWilsonPHD}, 
as a variation of Concurrent Pattern Calculus \cite{GivenWilsonGorlaJay10,GivenWilsonPHD},
and has a behavioural theory as a specialisation of \cite{GivenWilsonGorla13}.
Similarly, the language $\Lang_{S,M,C,I,B}$ is very similar to pattern-matching Spi calculus
\cite{Haack:2006:PS:1165126.1165127}
and Psi calculi \cite{BJPV11},
albeit without the assertions or the possibility of repeated binding names in patterns.
There are also similarities between $\Lang_{S,M,C,I,B}$ and the polyadic synchronous $\pi$-calculus of
\cite{Carbone:2003:EPP:941344.941346}, although the intensionality is limited to the
channel, i.e.~inputs and outputs of the form $\iap s x.P$ and $\oap s a.P$ respectively.
For the joining languages:
$\Lang_{A,P,C,\mathit{NO},J}$ represents Join Calculus \cite{Fournet_thereflexive};
and $\Lang_{S,P,C,\mathit{NO},J}$ the general rendezvous calculus \cite{Bocchi2004119},
and m-calculus \cite{DBLP:conf/popl/SchmittS03}, although the latter has higher order constructs and
other aspects that are not captured within the features here.

\begin{remark}
\label{rem:leq}
The languages $\Lambda_{s,a,m,p,\epsilon}$ can be easily ordered; in particular
$\Lang_{\alpha_1,\beta_1,\gamma_1,\delta_1,\epsilon_1}$ can be encoded into
$\Lang_{\alpha_2,\beta_2,\gamma_2,\delta_2,\epsilon_2}$ if it holds that
$\alpha_1\leq\alpha_2$ and
$\beta_1\leq\beta_2$ and
$\gamma_1\leq\gamma_2$ and
$\delta_1\leq\delta_2$ and
$\epsilon_1\leq \epsilon_2$, where $\leq$ is the least reflexive relation satisfying the following axioms:
\begin{equation*}
A \leq S\qquad\qquad M\leq P\qquad\qquad D\leq C\qquad\qquad
\mathit{NO}\leq \mathit{NM}\leq I \qquad\qquad B\leq J\; .
\end{equation*}
This can be understood as the lesser language variation being a special case of the more general language.
Asynchronous communication is synchronous communication with all outputs followed by $\zero$.
Monadic communication is polyadic communication with all tuples of arity one.
Dataspace-based communication is channel-based communication with all $k$-ary tuples communicating with channel name $k$.
All name-matching communication is intensional communication without any compounds,
and no-matching capability communication is both without any compounds and with only binding names
in patterns.
Lastly, binary communication is joining communication with all joining inputs having only a single input pattern.
\end{remark}

The operational semantics of the languages is given here via reductions as in
\cite{milner:polyadic-tutorial,Honda95onreduction-based,givenwilson:hal-01026301}.
An alternative style is via a {\em labelled transition system} (LTS) such as \cite{G:IC08}.
Here the reduction based style is to simplify having to define here the (potentially complex)
labels that occur when both intensionality and joining are in play. However, the LTS style can be
used for intensional languages \cite{BJPV11,GivenWilsonPHD,GivenWilsonGorla13}, and indeed
captures many\footnote
{Perhaps all of the binary languages here, although this has not been proven.}
of the languages here \cite{GivenWilsonGorla13}.
For the joining languages the techniques used in \cite{Fournet99bisimulationsin} can be used for
the no-matching joining languages, with the techniques of \cite{GivenWilsonGorla13} used to extend
intensionality\footnote{This has not been proven as yet, however there appears no reason it should not be
straightforward albeit very tedious.}.

Substitutions, denoted $\sigma,\rho,\ldots$, in non-pattern-matching and name-matching
languages are mappings (with finite domain) from names to names. For intensional languages
substitutions are mappings (also finite domain) from names to terms.
The application of a substitution $\sigma$ to a pattern $p$ is defined by: 
\begin{equation*}
\sigma x = \sigma(x)\ \ x\in\mbox{domain}(\sigma)\qquad
\sigma x = x\ \ x\not\in\mbox{domain}(\sigma)\qquad
\sigma \pro x = \pro {(\sigma x)}\qquad
\sigma (p\bullet q) = (\sigma p)\bullet(\sigma q)\; .
\end{equation*}
Where substitution is as usual on names, and on the understanding that the name-match
syntax can be applied to any term as follows
$\pro x \define \pro x$ and $\pro {(s\bullet t)} \define \pro s\bullet \pro t$.

Given a substitution $\sigma$ and a process $P$, denote with $\sigma P$ the
(capture-avoiding) application of $\sigma$ to $P$ that behaves in the usual manner.
Note that capture can always be avoided by exploiting $\alpha$-equivalence, which can
in turn be assumed \cite{UBN07}. 

\renewcommand{\match}[2]{\{#1/\!\!/#2\}}

Interaction between processes is handled by matching some terms $\wt t$ with
some patterns $\wt p$, and possibly also equivalence of channel-names.
This is handled in two parts. The first part is the {\em match} rule $\match t p$ of a single term $t$
with a single pattern $p$ to create a substitution $\sigma$.
That is defined as follows:
\begin{eqnarray*}
\begin{array}{rcl}
\match t x &\define& \{t/x\}\\
\match a {\pro a} &\define& \{\}
\end{array}&\qquad\qquad&
\begin{array}{rcl}
\match {s\bullet t} {p\bullet q} &\define& \match s p \cup \match t q\\
\match t p &\mbox{undefined}& \mbox{otherwise.}
\end{array}
\end{eqnarray*}
Any term $t$ can be matched with a binding name $x$ to generate a substitution from the
binding name to the term $\{t/x\}$.
A single name $a$ can be matched with a name-match for that name $\pro a$ to yield the
empty substitution.
A compound term $s\bullet t$ can be matched by a compound pattern $p\bullet q$ when
the components match to yield substitutions $\match s p=\sigma_1$ and $\match t q=\sigma_2$,
the resulting substitution is the unification of $\sigma_1$ and $\sigma_2$.
Observe that since patterns are well-formed, the substitutions of components will always have
disjoint domain.
Otherwise the match is undefined.

The second part is then the {\em poly-match} rule $\pmtch (\wt t ; \wt p)$ that
determines matching of a sequence of terms $\wt t$ with a sequence of patterns $\wt p$,
defined below.
\begin{equation*}
\pmtch(;)=\emptyset \qquad\qquad
\prooftree \match s p =\sigma_1 \qquad \pmtch (\wt t ; \wt q) =\sigma_2
\justifies \pmtch(s,\wt t;p,\wt q) = \sigma_1\cup \sigma_2
\endprooftree \; .
\end{equation*}
The empty sequence matches with the empty sequence to produce the empty substitution.
Otherwise when there is a sequence of terms $s,\wt t$ and a sequence of patterns $p,\wt q$,
the first elements are matched $\match s p$ and the remaining sequences use the poly-match rule.
If both are defined and yield substitutions, then the union of substitutions is the
result.
(Like the match rule, the union is ensured disjoint domain by well-formedness of inputs.)
Otherwise the poly-match rule is undefined, for example when a single match fails, or the
sequences are of unequal arity.

Interaction is now defined by the following axiom for the binary languages:
\begin{eqnarray*}
\oap s {\wt t}.P\bnf \iap s {\wt p}.Q &\quad\redar\quad& P\bnf\sigma Q
\qquad\qquad \pmtch(\wt t;\wt p)=\sigma
\end{eqnarray*}
and for the joining languages:
\begin{eqnarray*}
\oap {s_1} {\wt {t_1}}.P_1\bnf\ldots\bnf\oap {s_k} {\wt {t_k}}.P_k\bnf 
\join{\iap {s_1} {\wt {p_1}}\bnf\ldots\bnf\iap {s_k} {\wt {p_k}}}Q &\redar& P_1\bnf\ldots\bnf P_k\bnf\sigma Q
\qquad \sigma=\bigcup_{\{i\in 1\ldots k\}}\pmtch(\wt {t_i};\wt {p_i})\; .
\end{eqnarray*}
In both axioms, the $P$'s are omitted in the asynchronous languages,
and the $s$'s are omitted for the dataspace-based languages.
The axioms state that when the poly-match of the terms of the output(s) $\wt t$ match with the
input pattern(s) of the input $\wt p$ (and in the channel-based setting the output and input pattern(s)
are along the same channels) yields a substitution $\sigma$, then reduce to ($P$(s) in the synchronous languages
in parallel with) $\sigma$ applied to $Q$.

The general reduction relation $\redar$ includes the interaction axiom for the language in question
as well as the following three rules:
\begin{equation*}
\begin{array}{c}
\prooftree P\redar P'
\justifies P\bnf Q \redar P'\bnf Q
\endprooftree \qquad
\prooftree P\redar P'
\justifies \res a P \redar \res a P'
\endprooftree \qquad
\prooftree P\equiv Q\quad Q\redar Q'\quad Q'\equiv P'
\justifies P\redar P'
\endprooftree\quad\; .
\end{array}
\end{equation*}
The reflexive transitive closure of $\redar$ is denoted by $\Redar$.

Lastly, for each language let $\beq$ denote
a reduction-sensitive reference behavioural equivalence for that language, e.g.~a barbed equivalence.
For the non-intensional languages these are mostly already known, either by
their equivalent language in the literature, such as asynchronous/synchronous monadic/polyadic
$\pi$-calculus or Join Calculus, or from \cite{G:IC08}.
For the intensional languages the results in \cite{GivenWilsonGorla13} can be used.
For the other joining languages the techniques used in \cite{Fournet99bisimulationsin} can be used for
the no-matching joining languages, with the techniques of \cite{GivenWilsonGorla13} used to extend
intensionality\footnote{This has not been proven as yet, however there appears no reason it should not be
possible, and the results here rely upon the existence of an equivalence relation, not any particular one.}.

\section{Encodings}
\label{sec:encoding}

This section recalls the definition of valid encodings as well as some
useful theorems (details in \cite{G:CONCUR08}) for formally relating process calculi.
The validity of such criteria in developing expressiveness studies emerges from the
various works \cite{G:IC08,G:DC10,G:CONCUR08}, that have also recently inspired similar works
\cite{LPSS10,Lanese:2010:EPP:2175486.2175506,gla12}. 

An {\em encoding} of a language $\Lang_1$ into another language $\Lang_2$ is a pair
$(\encode\cdot,\renpol)$ where $\encode\cdot$ translates every $\Lang_1$-process into
an $\Lang_2$-process and $\renpol$ maps every name (of the source language) into a tuple
of $k$ names (of the target language), for $k > 0$.
The translation $\encode\cdot$ turns every term of the source language into a term of the
target; in doing this, the translation may fix some names to play a precise r\^ole 
or may translate a single name into a tuple of names. This can be obtained
by exploiting $\renpol$.

Now consider only encodings that satisfy the following properties.
Let a {\em $k$-ary context} $\context C {\cdot_1; \ldots; \cdot_k}$ be a term where $k$
occurrences of $\zero$ are linearly replaced by the holes $\{\cdot_1;
\ldots; \cdot_k\}$ (every one of the $k$ holes must occur once and only once).
Denote with $\redar_i$ and $\Redar_i$ 
the relations $\redar$ and $\Redar$ in language $\Lang_i$;
denote with $\redar^\omega_i$ an infinite sequence of reductions in $\Lang_i$.
Moreover, let $\beq_i$ denote the reference behavioural equivalence for language $\Lang_i$.
Also, let $P \suc_i$ mean that there exists $P'$ such that $P \Redar_i P'$ and $P' \equiv P''\bnf \ok$,
for some $P''$.
Finally, to simplify reading, let $S$ range
over processes of the source language (viz., $\Lang_1$) and $T$ range
over processes of the target language (viz., $\Lang_2$).

\begin{definition}[Valid Encoding]
\label{def:ve}
An encoding $(\encode\cdot,\renpol)$ of $\Lang_1$ into $\Lang_2$
is {\em valid} if it satisfies the following five properties:
\begin{enumerate}
\item {\em Compositionality:} for every $k$-ary operator $\op$ of $\Lang_1$
and for every subset of names $N$,
there exists a $k$-ary context $\CopN C \op N {\cdot_1; \ldots; \cdot_k}$ of $\Lang_2$
such that, for all $S_1,\ldots,S_k$ with ${\sf fn}(S_1,\ldots,S_k) = N$, it holds
that $\encode{\op(S_1,\ldots,S_k)} = \CopN C \op N {\encode{S_1};\ldots;\encode{S_k}}$.

\item {\em Name invariance:}
for every $S$ and substitution $\sigma$, it holds that
$\encode{\sigma S} = \sigma'\encode S$ if $\sigma $ is injective and
$\encode{\sigma S} \beq_2\ \sigma'\encode S$ otherwise
where $\sigma'$ is such that 
$\renpol(\sigma(a)) = \sigma'(\renpol(a))$
for every name $a$.

\item {\em Operational correspondence:}
\begin{itemize}
\item for all $S \Redar_1 S'$, it holds that $\encode S \Redar_2 \beq_2 \encode {S'}$;
\item for all $\encode S \Redar_2 T$, there exists $S'$ such that $S \Redar_1\!\! S'$ 
and $T \Redar_2 \beq_2\!\! \encode {S'}$.
\end{itemize}

\item {\em Divergence reflection:}
for every $S$ such that 
$\encode S \redar\!\!_2^\omega$, it holds that  $S$ \mbox{$\redar\!\!_1^\omega$}.

\item {\em Success sensitiveness:}
for every $S$, it holds that $S \suc_1$ if and only if $\encode S \suc_2$.
\end{enumerate}
\end{definition}

Now recall two results concerning valid encodings that are useful for later proofs.

\begin{proposition}[Proposition 5.5 from \cite{G:CONCUR08}]
\label{prop:deadlock}
Let $\encode\cdot$ be a valid encoding; then, $S \noredar\!\!_1$
implies that $\encode S \noredar\!\!_2$.
\end{proposition}

\begin{proposition}[Proposition 5.6 from \cite{G:CONCUR08}]
\label{prop:ctx_top}
Let $\encode\cdot$ be a valid encoding; then for every set of names $N$,
it holds that $\CopN C | N {\cdot_1,\cdot_2}$ has both its holes at top-level.
\end{proposition}

\section{Joining vs Binary}
\label{sec:join_only}

This section considers the expressive power gained by joining. It turns out that joining adds expressive
power that cannot be represented by binary languages regardless of other features.

The expressive power gained by joining can be captured by the concept of the {\em coordination degree} of a language $\Lang$, denoted $\cood\Lang$,
as the least upper bound on the number of processes that must coordinate to yield a reduction in $\Lang$.
For example, all the binary languages $\Lang_{-,-,-,-,B}$ have coordination degree 2 since their
reduction axiom is only defined for two processes.
By contrast, the coordination degree of the joining languages is $\infty$ since there is no bound on the
number of inputs that can be part of a join. 

\begin{theorem}\label{thm:int_deg_gt}
If $\cood{\Lang_1}>\cood{\Lang_2}$ then there exists no valid encoding $\encode\cdot$ from $\Lang_1$ into $\Lang_2$.
\end{theorem}
\begin{proof}
By contradiction, assume there is a valid encoding $\encode\cdot$.
Pick $i$ processes $S_1$ to $S_i$ where $i=\cood{\Lang_2}+1$ such that all these processes
must coordinate to yield a reduction and yield success.
That is, $S_1\bnf \ldots \bnf S_i\redar\ok$
but not if any $S_j$ (for $1\leq j\leq i$) is replaced by the null process $\zero$.
By validity of the encoding it must be that $\encode{S_1\bnf \ldots \bnf S_i}\redar$
and $\encode{S_1\bnf \ldots \bnf S_i}\suc$.

By compositionality of the encoding 
$\encode{S_1\bnf \ldots \bnf S_i}=\mathcal{C}_S$ where
$\mathcal{C}_S$ must be of the form\linebreak
$\context{C^N_|}{\encode{S_1},\context{C^N_|}{\ldots,\context{C^N_|}{\encode{S_{i-1}},\encode{S_i}}\ldots}}$.
Now consider the reduction $\encode{S_1\bnf \ldots \bnf S_i}\redar$ that can be at most between $i-1$
processes by the coordination degree of $\Lang_2$.
If the reduction does {\em not} involve some process $\encode {S_j}$ then it follows that
$\encode{S_1\bnf\ldots\bnf S_{j-1}\bnf\zero\bnf S_{j+1}\bnf\ldots\bnf S_i}\redar$
(by replacing the $\encode{S_j}$ in the context $\mathcal{C}_S$ with $\encode\zero$).
By construction of $S_1\bnf \ldots\bnf S_i$ and $\cood{\Lang_2}<i$ there must
exist some such $S_j$.
However, this contradicts the validity of the encoding since
$S_1\bnf\ldots\bnf S_{j-1}\bnf\zero\bnf S_{j+1}\bnf\ldots\bnf S_i\not\redar$.
The only other possibility is if $\encode{S_j}$ blocks the reduction by blocking some $\encode{S_k}$.
This can only occur when $\encode{S_k}$ is either underneath an interaction primitive
(e.g.~$\oap s {\wt t}.\encode{S_k}$) or inside a conditional (e.g.~$\ift s t {\encode{S_k}}$).
Both require that $\encode{S_k}$ not be top level in $\mathcal{C}_S$, which can be proven contradictory
by $i-1$ applications of Proposition~\ref{prop:ctx_top}.
\end{proof}

\begin{corollary}
\label{cor:join_gt_bin}
There exists no valid encoding from $\Lang_{-,-,-,-,J}$ into $\Lang_{-,-,-,-,B}$.
\end{corollary}

In the other direction the result is ensured by Remark~\ref{rem:leq}.
Thus for any two languages which differ only by one being binary and the other joining, the
joining language is strictly more expressive than the binary language.

\section{Joining and Synchronicity}
\label{sec:join_synch}

This section considers the relation between joining and synchronicity. It turns out that
the two are orthogonal and do not influence the other's expressiveness.

It is sufficient to consider the languages $\Lang_{A,M,D,\mathit{NO},J}$ and
$\Lang_{A,P,D,\mathit{NO},J}$ and $\Lang_{A,M,D,\mathit{NM},J}$.
The other asynchronous joining languages can encode their synchronous joining
counterparts in the usual manner \cite{honda1991object}. 
For example, the encoding from $\Lang_{S,M,C,\mathit{NO},B}$ into $\Lang_{A,M,C,\mathit{NO},B}$
given by
\begin{eqnarray*}
\encode{\oap n a . P} &\define& \res z (\oap n z \bnf \iap z x.    (\oap x a \bnf \encode{P}))\\
\encode{\iap n a . Q} &\define& \res x  \iap n z.(    \oap z x \bnf \iap x a . \encode{Q})
\end{eqnarray*}
can be adapted in the obvious manner for
$\Lang_{S,M,C,\mathit{NO},J}$ into $\Lang_{A,M,C,\mathit{NO},J}$ as follows
\begin{eqnarray*}
\encode{\oap n a . P} &\define& \res z (\oap n z \bnf \join {\iap z x}    (\oap x a \bnf \encode{P}))\\
\encode{\join{\iap {n_1} {a_1}\bnf\ldots\bnf \iap {n_i} {a_i}} Q} &\define&
\res {x_1,\ldots,x_i}  \join {\iap {n_1} {z_1}\bnf \ldots\bnf \iap {n_i} {z_i}}\\
&&\qquad(    \oap {z_1} {x_1} \bnf \ldots \bnf \oap {z_i} {x_i} \bnf \join
  {\iap {x_1} {a_1}\bnf\ldots\bnf \iap {x_i} {a_i} } \encode{Q})\; .
\end{eqnarray*}
The idea for binary languages is that the encoded output creates a fresh name $z$ and sends it to the encoded input.
The encoded input creates a fresh name $x$ and sends it to the encoded output along channel name $z$.
The encoded output now knows it has communicated and evolves to $\encode P$ in parallel with the
original $a$ sent to the encoded input along channel name $x$.
When the encoded input receives this it can evolve to $\encode Q$.
The joining version is similar except the join synchronises with all the encoded outputs at once,
sends the fresh names $x_j$ in parallel, and then synchronises on all the $a_j$ in the last step.

The encoding above is shown for $\Lang_{S,M,C,\mathit{NO},J}$ into $\Lang_{A,M,C,\mathit{NO},J}$
and is the identity on all other process forms. This can be proven to be a valid encoding.

\begin{lemma}
\label{lem:minussynch-match}
Given a $\Lang_{S,M,C,\mathit{NO},J}$ input $P$
and output $Q$
then $\xtrans P \bnf \xtrans Q \redar$ if and only if $P\bnf Q \redar$.
\end{lemma}

\begin{lemma}
\label{lem:minussynch-equiv}
If $P\equiv Q$ then $\xtrans P \equiv \xtrans Q$.
Conversely, if $\xtrans P \equiv Q$ then $Q=\xtrans {P'}$ for some $P'\equiv P$.
\end{lemma}

\begin{lemma}
\label{lem:minussynch-red}
The translation $\xtrans\cdot$ from $\Lang_{S,M,C,\mathit{NO},J}$ into $\Lang_{A,M,C,\mathit{NO},J}$ preserves
and reflects reductions.
\end{lemma}

\begin{theorem}
\label{thm:synch}
There is a valid encoding from $\Lang_{S,M,C,\mathit{NO},J}$ into $\Lang_{A,M,C,\mathit{NO},J}$.
\end{theorem}
\begin{proof}
Compositionality and name invariance hold by construction.
Operational correspondence (with structural equivalence in the place of $\beq$)
and divergence reflection follow from Lemma~\ref{lem:minussynch-red}.
Success sensitiveness can be proved as follows: $P\suc$ means that there exists $P'$ and
$k\geq 0$ such that $P\redar^k P'\equiv P''\bnf \ok$; by exploiting Lemma~\ref{lem:minussynch-red}
$k$ times and Lemma~\ref{lem:minussynch-equiv} obtain that
$\xtrans P \redar^j \xtrans {P'}\equiv \xtrans {P''}\bnf \ok$ where $j$ can be determined from the
instantiations of Lemma~\ref{lem:minussynch-equiv},
i.e.~that $\xtrans P \suc$.
The converse implication can be proved similarly.
\end{proof}

\begin{corollary}
If there exists a valid encoding from $\Lang_{S,\beta,\gamma,\delta,B}$ into $\Lang_{A,\beta,\gamma,\delta,B}$
then there exists a valid encoding from $\Lang_{S,\beta,\gamma,\delta,J}$ into $\Lang_{A,\beta,\gamma,\delta,J}$.
\end{corollary}
\begin{proof}
Theorem~\ref{thm:synch} applies directly for all channel-based languages.
The only other cases can encode channels and so use encodings of the channel-based solution above.
For the polyadic and name-matching languages this holds by Proposition~4.1 of \cite{G:IC08},
otherwise for the intensional languages this holds by Theorem~6.4 of \cite{givenwilson:hal-01026301}.
\end{proof}

The following results complete the formalisation that coordination is orthogonal to synchronicity.

\begin{theorem}
\label{thm:no_synch_2_asynch}
There exists no valid encoding from $\Lang_{S,M,D,\mathit{NM},J}$ into $\Lang_{A,M,D,\mathit{NM},J}$.
\end{theorem}
\begin{proof}
The proof is by contradiction. Consider two processes
$P=\join{\iap {} x}\ift x b \Omega$ (where $\Omega$ is a divergent process) and $Q=\oap {} a . Q'$.
Since $P\bnf Q\redar$ by validity of the encoding $\encode{P\bnf Q}\redar$ and this must be
between some $R_1=\oap {} m$ for some $m$ and $R_2$. 
Observe that $R_1\bnf R_2$ cannot be a reduct of either $\encode {P}$ or $\encode{Q}$ since then
either $P$ or $Q$ would reduce and this contradicts Proposition~\ref{prop:deadlock}.

If $R_1$ arises from $\encode P$ then it can be shown that $\encode P$ must also include a top level
join since otherwise there would be no join in $\encode P$ that can bind some name to $x$ and
name invariance or divergence reflection would be shown to fail
(i.e.~$P\bnf Q\redar\ift a b\Omega\bnf Q'$ and $\{b/a\}\ift a b\Omega\bnf Q'\redar\!\!^\omega$ while
$\context{C^N_|}{\encode{P},\encode {Q}}\Redar$ does no inputs on any part of $\encode P$
and so must always or never diverge regardless of interaction with $\encode Q$).
Thus $\encode P$ must include a top level join and further it must include an input pattern
$\iap {} {\pro n}$ for some $n\neq m$ since otherwise if the join was only
$\join{\iap {} {z_1}\bnf\ldots\bnf\iap {}{z_i}}R'$ for some $\wt z$ and $R'$ then
$\encode{P\bnf \ldots\bnf P}$ for $i$ instances of $P$ would reduce while $P\bnf \ldots\bnf P$
does not contradicting Proposition~\ref{prop:deadlock}.
It follows that $\encode Q$ must include $\iap {} {\pro m}$ as part of some join under which there
must be an output that is able to send at least one name to $\encode P$ via an output $\oap {} {d}$ for some $d$
(this could be any number of names, but assume 1 here for simplicity).
Now consider the name $d$.
If $d=m$ then $\encode P \redar$ and this contradicts validity of the encoding since $P\not\redar$.
If $d=n$ then $n$ is not bound in $\encode P$ and so it can be shown that either: this fails name invariance
or divergence reflection
(again by $P\bnf Q\redar\ift a b\Omega\bnf Q'$ and $\{b/a\}\ift a b\Omega\bnf Q'\redar\!\!^\omega$);
or there must be a further input in $\encode P$ that is binding as in the next case.
If $d\neq m$ and $d\neq n$ then it can be shown that $\encode{P\bnf Q\bnf P}$ can reduce such that the
input under consideration interacts with the $\oap {} m$ from the other $\encode P$
and this ends up contradicting operational correspondence.

If $R_1$ arises from $\encode Q$ then it can be shown that $\encode Q$ must also include a top level
join since otherwise when $Q'=\Omega$ then $\encode Q$ would always diverge or never diverge 
regardless of interaction with $\encode P$ and this contradicts divergence reflection.
Thus $\encode Q$ must include a top level join and further it must include an input pattern
$\iap {} {\pro n}$ for some $n\neq m$ since otherwise if the join was only
$\join{\iap {} {z_1}\bnf\ldots\bnf\iap {}{z_i}}R'$ for some $\wt z$ and $R'$ then
$\encode{Q\bnf \ldots\bnf Q}$ for $i$ instances of $Q$ would reduce while $Q\bnf \ldots\bnf Q$
does not contradicting Proposition~\ref{prop:deadlock}.
Now consider when $Q'=\ift a b \ok$ and the substitution $\sigma=\{b/a\}$.
Clearly $P\bnf \sigma Q\bnf Q\redar S$ where either:
$S\redar\!\!^\omega$ and $S\suc$; or $S\not\redar\!\!^\omega$ and $S\not\suc$.
However it can be shown that the top level join in $\encode Q$ is not able to discriminate and
thus that there exist two possible reductions $\encode{P\bnf \sigma Q\bnf Q}\redar R'$ to an $R'$
where either:
$R'\redar\!\!^\omega$ and $R'\not\suc$; or $R'\not\redar\!\!^\omega$ and $R\suc$;
both of which contradict divergence reflection and success sensitiveness.
\end{proof}

\begin{theorem}
There exists no valid encoding from $\Lang_{S,\beta,D,\mathit{NO},J}$ into $\Lang_{A,\beta,D,\mathit{NO},J}$.
\end{theorem}
\begin{proof}
This is proved in the same manner as Theorem~\ref{thm:no_synch_2_asynch}.
\end{proof}

That joining does not allow for an encoding of synchronous communication alone is not surprising,
since there is no control in the input of which outputs are interacted with (without some other
control such as channel names or pattern-matching). Thus, being able to consume more outputs in
a single interaction does not capture synchronous behaviours.

\section{Joining and Arity}
\label{sec:join_arity}

This section considers the relation between joining and arity. It turns out that these are
orthogonal.
Although there appear to be some similarities in that both have a base case (monadic or binary),
and an unbounded case (polyadic or joining, respectively), these cannot be used to encode one-another.
This is captured by the following result.

\begin{theorem}
\label{thm:no_poly_2_join}
There exists no valid encoding from $\Lang_{A,P,D,\mathit{NO},B}$ into $\Lang_{A,M,D,\mathit{NO},J}$.
\end{theorem}
\begin{proof}
The proof is by contradiction, assume there exists a valid encoding $\encode\cdot$.
Consider the $\Lang_{A,P,D,\mathit{NO},B}$ processes $P=\oap {} {a,b}$ and $Q=\iap {} {x,y} .\ok$.
Clearly it holds that $P\bnf Q\redar\ok$ and so $\encode{P\bnf Q}\redar$ and $\encode{P\bnf Q}\suc$
by validity of the encoding. Now consider the reduction $\encode{P\bnf Q}\redar$.

The reduction must be of the form
$\oap {} {m_1} \bnf \ldots \bnf \oap {} {m_i} \bnf \join {\iap {} {z_1}\bnf\ldots\bnf\iap {} {z_i}} T'$
for some $\wt m$ and $\wt z$ and $i$ and $T'$.
Now consider the process whose encoding produces $\join {\iap {} {z_1}\bnf\ldots\bnf\iap {} {z_i}} T'$,
assume $Q$ although the results do not rely on this assumption.
If any $\oap {} {m_j}$ are also from the encoding of $Q$ then it follows that the encoding of
$i$ instances of $Q$ in parallel will reduce, i.e.~$\encode{Q\bnf \ldots\bnf Q}\redar$, while
$Q\bnf \ldots \bnf Q\not\redar$.
Now consider two fresh processes $S$ and $T$ such that $S\bnf T\redar$ with some arity that is not $2$
and $S\not\redar$ and $T\not\redar$.
It follows that $\encode{S\bnf T}\redar$ (and $\encode S \not\redar$ and $\encode T \not\redar$)
and $\encode {S\bnf T}$ must include at least one $\oap {} n$ to do so.
This $\oap {} n$ must arise from either $\encode{S}$ or $\encode {T}$, and conclude by showing that
the encoding of $i$ instances of either $S$ or $T$ in parallel with $Q$ reduces, while the
un-encoded processes do not.
\end{proof}

\begin{corollary}
If there exists no valid encoding from
$\Lang_{\alpha,P,\gamma,\delta,B}$ into $\Lang_{\alpha,M,\gamma,\delta,B}$, then there exists
no valid encoding from $\Lang_{\alpha,P,\gamma,\delta,-}$ into $\Lang_{\alpha,M,\gamma,\delta,J}$.
\end{corollary}
\begin{proof}
The technique in Theorem~\ref{thm:no_poly_2_join} applies to all dataspace-based no-matching languages.
Dataspace-based name-matching languages build upon Theorem~\ref{thm:no_poly_2_join} with
$Q=\iap {} {x,y}.\ift a x \ok$ to then ensure that binding occurs and not only name-matching,
the proof is concluded via contradiction of name invariance and success sensitiveness as in
Theorem~\ref{thm:no_synch_2_asynch}.
For the channel-based communication it is easier to refer to Theorem~\ref{thm:no_chan_2_join} to
illustrate that this is not possible than to extend the proof above.
\end{proof}

Thus joining does not allow for encoding polyadicity in a monadic language unless it could already be
encoded by some other means.
In the other direction, the inability to encode joining into a binary language is already 
ensured by Corollary~\ref{cor:join_gt_bin}.

\section{Joining and Communication Medium}
\label{sec:join_comm}

This section considers the relation between joining and communication medium. Again joining turns out
to be orthogonal to communication medium and neither can encode the other.
The key to this is captured in the following result.

\begin{theorem}
\label{thm:no_chan_2_join}
There exists no valid encoding from $\Lang_{A,M,C,\mathit{NO},B}$ into $\Lang_{A,M,D,\mathit{NO},J}$.
\end{theorem}
\begin{proof}
The proof is by contradiction, assume there exists a valid encoding $\encode\cdot$.
Consider the $\Lang_{A,M,C,\mathit{NO},B}$ processes $P=\oap {a} {b}$ and $Q=\iap {a} {x} .\ok$.
Clearly it holds that $P\bnf Q\redar\ok$ and so $\encode{P\bnf Q}\redar$ and $\encode{P\bnf Q}\suc$
by validity of the encoding. Now consider the reduction $\encode{P\bnf Q}\redar$.

The reduction must be of the form
$\oap {} {m_1} \bnf \ldots \bnf \oap {} {m_i} \bnf \join {\iap {} {z_1}\bnf\ldots\bnf\iap {} {z_i}} T'$
for some $\wt m$ and $\wt z$ and $i$ and $T'$.
Now consider the process whose encoding produces $\join {\iap {} {z_1}\bnf\ldots\bnf\iap {} {z_i}} T'$,
assume $Q$ although the results do not rely on this assumption.
If any $\oap {} {m_j}$ are also from the encoding of $Q$ then it follows that the encoding of
$i$ instances of $Q$ in parallel will reduce, i.e.~$\encode{Q\bnf \ldots\bnf Q}\redar$, while
$Q\bnf \ldots \bnf Q\not\redar$.
Now consider two fresh processes $S=\oap c d$ and $T=\iap c z.\zero$.
Since $S\bnf T\redar$ it follows that $\encode{S\bnf T}\redar$ and must include at least one
$\oap {} n$ to do so.
This $\oap {} n$ must arise from either $\encode{S}$ or $\encode {T}$, and conclude by showing that
the encoding of $i$ instances of either $S$ or $T$ in parallel with $Q$ reduces, while the
un-encoded processes do not.
\end{proof}

\begin{corollary}
If there exists no valid encoding from
$\Lang_{\alpha,\beta,C,\delta,B}$ into $\Lang_{\alpha,\beta,D,\delta,B}$, then there exists
no valid encoding from $\Lang_{\alpha,\beta,C,\delta,-}$ into $\Lang_{\alpha,\beta,D,\delta,J}$.
\end{corollary}
\begin{proof}
The technique in Theorem~\ref{thm:no_chan_2_join} applies to all monadic languages
(the addition of name-matching can be proved using the techniques as in Theorem~\ref{thm:no_synch_2_asynch}).
For the polyadic no-matching setting the result above holds by observing that the arity must
remain fixed for an encoding, i.e.~$\encode {\oap a {b_1,\ldots,b_i}}$ is\linebreak
encoded to inputs/outputs all of
some arity $j$. If the arity is not uniform then the encoding fails\linebreak
either operational correspondence
(i.e.~$\encode {\iap a x.\zero\bnf \oap a {b_1,b_2}}\redar$)
or divergence reflection as in sub-case (2) of Theorem~\ref{thm:no_int_2_join} except here with
arity instead of number of names.
\end{proof}

Thus joining does not allow for encoding channels in a dataspace-based language unless it could already be
encoded by some other means.
In the other direction, the inability to encode joining into a binary language is already 
ensured by Corollary~\ref{cor:join_gt_bin}.

\section{Joining and Pattern-Matching}
\label{sec:join_pattern}

This section considers the relations between joining and pattern-matching.
The great expressive power of name matching \cite{G:IC08} and intensionality \cite{givenwilson:hal-01026301}
prove impossible to encode with joining.
In the other direction, joining cannot be encoded by any form of pattern-matching.

The first result is to prove that intensionality cannot be encoded by joining.
Recall that since intensionality alone can encode all other features aside from joining,
it is sufficient to consider $\Lang_{A,M,D,I,B}$.

\begin{theorem}
\label{thm:no_int_2_join}
There exists no valid encoding from $\Lang_{A,M,D,I,B}$ into $\Lang_{-,-,-,\delta,J}$ where $\delta\neq I$.
\end{theorem}
\begin{proof}
The proof is by contradiction and similar to Theorem~7.1 of \cite{givenwilson:hal-01026301}.
Assume there exists a valid encoding $\xtrans\cdot$ from $\Lang_{A,M,D,I,B}$ into
$\Lang_{\alpha,\beta,\gamma,\delta,J}$ for some $\alpha$ and $\beta$ and $\gamma$ and $\delta$
where $\delta\neq I$.
Consider the encoding of the processes $S_0=\join{\iap {} x} \oap {} m$ and $S_1=\oap {} a$.
Clearly $\xtrans{S_0 \bnf S_1}\redar$  since $S_0\bnf S_1\redar$.
There exists a reduction $\xtrans{S_0\bnf S_1}\redar$ that must be between a
join and some outputs that have combined maximal arity $k$.
(The combined arity is the sum of the arities of all the input-patterns of the join involved,
e.g.~$\join{\iap {} {a,b}\bnf\iap {} {c}} \zero$ has combined arity 3.)

Now define the following processes 
$S_2 \define \oap {} {a_1\bullet\ldots\bullet a_{2k+1}}$
and
$S_3 \define \iap {} {\pro{a_1}\bullet\ldots\bullet\pro{a_{2k+1}}} . \oap {} m$
where $S_2$ outputs $2k+1$ distinct names in a single term, and $S_3$ matches all of these names in
a single intensional pattern.
Since $S_2\bnf S_0\redar$ it must be that $\xtrans{S_2\bnf S_0}\redar$
for the encoding to be valid.
Now consider the maximal combined arity of the reduction $\xtrans{S_2\bnf S_0}\redar$.
\begin{itemize}
\item If the arity is $k$ 
  consider the reduction $\xtrans{S_2\bnf S_3}\redar$ with the combined maximal arity $j$ which
  must exist since $S_2\bnf S_3\redar$.
  Now consider the relationship of $j$ and $k$.
  \begin{enumerate}
  \item If $j=k$ 
    then the upper bound on the number of names that are matched in the
    reduction is $2k$ (when each name is matched via a distinct channel).
    Since not all $2k+1$ tuples of names from $\renpol (a_i)$ can be matched in the reduction then
    there must be at least one tuple $\renpol (a_i)$ for $i\in\{1,\ldots,2k+1\}$ that is not being
    matched in the interaction $\xtrans{S_2 \bnf S_3}\redar $.
    Now construct $S_4$ that differs from $S_3$ only by swapping one such name $a_i$ with
    $m$: $
    S_4 \define \iap {} {\pro{a_1}\bullet\ldots\pro{a_{i-1}}\bullet\pro m
    \bullet\pro{a_{i+1}}\ldots\pro{a_{k+2}}}. \oap {} {a_i}$.
    Now consider the context $\context {C^N_|} {\encode{S_2},\encode\cdot}=\encode{S_2\bnf \cdot}$
    where $N=\{\wt a\cup m\}$.
    Clearly neither $\context {C^N_|} {\encode{S_2},\encode{\zero}}\redar$ nor
    $\context {C^N_|} {\encode{S_2},\encode{S_4}}\redar$ as this would contradict
    Proposition~\ref{prop:deadlock}.
    However, since $S_3$ and $S_4$ differ only by the position of one name whose tuple $\renpol(\cdot)$
    does not appear in the
    reduction $\encode{S_2\bnf S_3}\redar$, it follows that the reason 
    $\context {C^N_|} {\encode{S_2},\encode{S_4}}\not\redar$ must be due to a structural
    congruence difference between $\context {C^N_|} {\encode{S_2},\encode{S_3}}$
    and $\context {C^N_|} {\encode{S_2},\encode{S_4}}$.
    Further, by compositionality of the encoding the difference can only be between
    $\encode{S_3}$ and $\encode{S_4}$.
    Since Proposition~\ref{prop:deadlock} ensures that $\encode{S_3}\not\redar$
    and $\encode{S_4}\not\redar$, the only possibility is a structural difference
    between $\encode{S_3}$ and $\encode {S_4}$.
    Now exploiting 
    $\sigma=\{m/a_i,a_i/m\}$
    such that $\sigma S_4=S_3$
    yields contradiction.
%
%
  \item If $j\neq k$ then obtain that \encode{$S_2}$ must be able to interact with both
    combined arity $k$ and combined arity $j$.
    That is, $\xtrans{S_2\bnf\cdot}=\context{C^N_|}{\xtrans {S_2},\xtrans\cdot}$
    where $N=\{\wt a\cup m\}$ and that
    $\context{C^N_|}{\xtrans {S_2},\xtrans{S_0}}$ reduces with combined arity $k$ and
    $\context{C^N_|}{\xtrans {S_2},\xtrans{S_3}}$ reduces with combined arity $j$.
    Now it is straightforward, if tedious, to show that since $S_0\bnf S_3\not\redar$
    that $\context{C^N_|}{\xtrans {S_2},\xtrans{S_0\bnf S_3}}$ can perform the same initial
    reductions as either
    $\context{C^N_|}{\xtrans {S_2},\xtrans{S_0\bnf \zero}}$ or 
    $\context{C^N_|}{\xtrans {S_2},\xtrans{\zero\bnf S_3}}$
    by exploiting operational correspondence and Proposition~\ref{prop:deadlock}.
    Thus, it can be shown that $\context{C^N_|}{\xtrans {S_2},\xtrans{S_0\bnf S_3}}$ can perform both
    the $k$ combined arity reduction of $\xtrans{S_2\bnf S_0}\redar$ and
    the $j$ combined arity reduction of $\xtrans{S_2\bnf S_3}\redar$.
    Now by exploiting the structural congruence rules it follows that neither of these
    initial reductions can prevent the other occurring.
    Thus,
    $\context{C^N_|}{\xtrans {S_2},\xtrans{S_0\bnf S_3}}$ must be able to do both of these 
    initial reductions in any order.
    Now consider the process $R$ that has performed both of these initial reductions.
    By operational correspondence it must be that $R\not\Redar\beq\encode{\oap {} m\bnf \oap {} m}$ since
    $S_2\bnf S_0\bnf S_3\not\Redar \oap {} m\bnf \oap {} m$.
    Therefore, $R$ must be able to roll-back the initial step with combined arity $j$;
    i.e~reduce to a state that is equivalent to the reduction not occurring.
    (Or the initial step with arity $k$, but either one is sufficient as
    by operational correspondence $R\Redar\beq\encode{\oap {} m\bnf S_3}$.)
    Now consider how many names are being matched in the initial reduction with combined arity $j$.
    If $j< k$ the technique of differing on one name used in the case of $j=k$ can be used
    to show that this would introduce divergence on the potential roll-back and thus contradict
    a valid encoding.
    Therefore it must be that $j > k$.
    Finally, by exploiting name invariance and substitutions like
    $\{(b_1\bullet\ldots\bullet b_{j})/a_1\}$ applied to $S_2$ and $S_3$
    it follows that either $j> k + j$ or both $S_2$ and $S_3$ must have
    infinitely many initial reductions which yields divergence.
  \end{enumerate}
\item If the combined arity is not $k$ then proceed like the second case above.
\end{itemize}
\vspace{-0.6cm}
\end{proof}

\begin{corollary}
If there exists no valid encoding from
$\Lang_{\alpha,\beta,\gamma,I,B}$ into $\Lang_{\alpha,\beta,\gamma,\delta,B}$, then there exists
no valid encoding from $\Lang_{\alpha,\beta,\gamma,I,-}$ into $\Lang_{\alpha,\beta,\gamma,\delta,J}$.
\end{corollary}

It follows that joining cannot represent intensionality in a language that does not have intensionality
already (including name-matching or no-matching languages).

The next result shows that name matching is insufficient to encode joining.

\begin{theorem}
\label{thm:no_name_2_join}
There exists no valid encoding from $\Lang_{A,M,D,\mathit{NM},B}$ into
$\Lang_{\alpha,\beta,\gamma,\mathit{NO},J}$.
\end{theorem}
\begin{proof}
The proof is by contradiction, assume there exists a valid encoding $\encode\cdot$.
Consider the $\Lang_{A,M,D,\mathit{NM},B}$ processes $P=\oap {} {a}$ and
$Q=\iap {} {\pro a} .(\oap {} b \bnf \ok)$.
Clearly it holds that $P\bnf Q\redar$ and $P\bnf Q\suc$ and so
$\encode{P\bnf Q}\redar$ and $\encode{P\bnf Q}\suc$ by validity of the encoding.
Now consider $\gamma$.
\begin{itemize}
\item If $\gamma=D$ then consider 
the substitution $\sigma=\{a/b,b/a\}$, it is clear that $P\bnf \sigma Q\not\redar$
and so $\encode{P\bnf \sigma Q}\not\redar$, however the only possibility that this holds is
when $\encode {\sigma Q}$ is blocked from interacting.
It is then straightforward if tedious to show that any such blocking of reduction would either
imply $\encode {\sigma(P\bnf Q)}\not\redar$ or $\sigma(P\bnf Q)\not\redar$
and thus contradict the validity of the encoding.

%
\item Otherwise it must be that $\gamma=C$. Now consider the reduction $\encode{P\bnf Q}\redar$
that must be of the form
$\oap {c_1} {\wt {m_1}} \bnf \ldots \bnf \oap {c_i} {\wt {m_i}} \bnf
\join {\iap {c_1} {\wt {z_1}}\bnf\ldots\bnf\iap {c_i} {\wt {z_i}}} T_1$
for some $\wt c$ and $\wt m$ and $\wt z$ and $i$ and $T_1$.
Again consider the substitution $\sigma=\{a/b,b/a\}$, it is clear that $\sigma P\bnf Q\not\redar$
and so $\encode{\sigma P\bnf Q}\not\redar$. The only way this can occur without contradicting
the validity of the encoding (as in the previous case) is when there is at least one $c_k$ in the
domain of some $\sigma '$ where $\sigma '(c_k)\neq c_k$ and $\encode{\sigma P}\beq \sigma'\encode{P}$
by definition of the encoding.
Now consider the process $S=\iap {} x.S'$,
clearly $P\bnf S\redar$ and so $\encode{P\bnf S}\redar$ as well.
The reduction $\encode{P\bnf S}\redar$ must be from the form
$\oap {d_1} {\wt {n_1}} \bnf \ldots \bnf \oap {d_j} {\wt {n_j}} \bnf
\join {\iap {d_1} {\wt {w_1}}\bnf\ldots\bnf\iap {d_j} {\wt {w_j}}} T_2$
for some $\wt d$ and $\wt n$ and $\wt w$ and $j$ and $T_2$.
Now if $i=j$ it follows that for each $k\in\{1\ldots i\}$ then $c_k=d_k$.
However, this contradicts the validity of the encoding since there is some $c_k$ in the domain of $\sigma '$
such that $\sigma '(c_k)\neq c_k$ and $\sigma P\bnf S\redar$ while $\encode {\sigma P\bnf S}\not\redar$.
Otherwise it must be that $i>j$ (otherwise if $i<j$ then $\encode {P\bnf S}\not\redar$)
and that $c_k\in\{c_{j+1},\ldots,c_i\}$.
Now consider when $S'=\ift x a \Omega$, 
clearly $P\bnf S\redar\equiv\Omega$ and $\sigma P\bnf S\redar\equiv \zero$ and so
$\encode{P\bnf S}$ diverges and $\encode{\sigma P\bnf S}\Redar\beq \zero$.
Now it can be shown that $P\bnf \sigma P\bnf S\bnf Q\redar\redar \equiv\ok$ while
$\encode{P\bnf \sigma P\bnf S\bnf Q}\suc$ and diverges since $\encode {\sigma P}$
can satisfy the first $j$ input patterns of $\encode{Q}$ and $\encode{\sigma P}$ the remaining
$i-j$, leaving the first $j$ input patterns of $\encode {P}$ to interact with $\encode S$
and yield divergence.
The only other possibility is that $\encode{P\bnf \sigma P\bnf S\bnf Q}\not\suc$.
However, this requires that $T_1$ check some binding name in $\wt z$ for equality with
$a$ before yielding success (i.e.~$\ift {z_1} a \ok$). This can in turn be shown
to contradict the validity of the encoding by adding another instance of $P$.
\end{itemize}
\vspace*{-0.6cm}
\end{proof}

\begin{corollary}
If there exists no valid encoding from
$\Lang_{\alpha,\beta,\gamma,\mathit{NM},B}$ into $\Lang_{\alpha,\beta,\gamma,\delta,B}$, then there exists
no valid encoding from $\Lang_{\alpha,\beta,\gamma,\mathit{NM},-}$ into $\Lang_{\alpha,\beta,\gamma,\delta,J}$.
\end{corollary}

Thus joining does not allow for encoding name-matching into a no-matching language unless it could already be
encoded by some other means.
In the other direction, the inability to encode joining into a binary language is already 
ensured by Corollary~\ref{cor:join_gt_bin}.

\section{Conclusions and Future Work}
\label{sec:conclude}

Languages with non-binary coordination have been considered before, although
less often than binary languages.
It turns out that increases in coordination degree correspond to increases in expressive power.
For example, an intensional binary language
cannot be encoded by a non-intensional joining language.
However, encodings from lower coordination degree languages into higher coordination degree
languages are still dependent upon other features.

This formalises that languages like the Join Calculus, general rendezvous calculus, and m-calculus
cannot be validly encoded into binary languages, regardless of other features.
Although there exist encodings from (for example) Join Calculus into $\pi$-calculus
\cite{Fournet_thereflexive}
these do not meet the criteria for a {\em valid encoding} used here.
The general approach used in such encodings is to encode joins by
$\encode{\join {\iap m x \bnf \iap n y} P}=\iap m x .\iap n y.\encode{P}$,
however this can easily fail operational correspondence, divergence reflection, or success sensitivity.
For example
${\join{\iap {c_1} {w}\bnf \iap {c_2} {x}}\ok \bnf \join {\iap {c_2} {y} \bnf \iap {c_1} {z}}\Omega
\bnf \oap {c_1} a\bnf \oap {c_2} b}$ will either report success or diverge, but its encoding can deadlock.
Even ordering the channel names to prevent this can be shown to fail under substitutions.

Future work along this line can consider coordination not merely to be binary or joining.
Indeed, a {\em splitting} language could be one where several output terms can be combined into
a split $\join {\oap m a\bnf \oap n b} P$ while inputs remain of the form $\iap m x .Q$.
Further, languages could support both joining and splitting primitives for full coordination.

\paragraph{Acknowledgments.} We would like to thank the reviewers for their constructive and helpful criticism.

\bibliographystyle{eptcs}
\bibliography{local}

\end{document}